\documentclass[12pt]{article}
\usepackage[english]{babel}

\usepackage{fullpage}
\usepackage{cite}

\usepackage{newpxtext,newpxmath}

\let\coloneqq\relax
\let\eqqcolon\relax

\usepackage[utf8]{inputenc}
\usepackage{amsthm}
\usepackage{amssymb}
\usepackage{amsmath}
\usepackage{bbold}
\usepackage{bbm}
\usepackage[pdftex, backref=page]{hyperref}
\usepackage{braket}
\usepackage{dsfont}
\usepackage{mathdots}
\usepackage{mathtools}
\usepackage{enumerate}
\usepackage[shortlabels]{enumitem}
\usepackage{csquotes}
\usepackage{stmaryrd}
\usepackage[cal=boondox]{mathalfa}
\usepackage{graphicx}
\usepackage{stackengine}
\usepackage{scalerel}
\usepackage{tensor}       
\usepackage{array}
\usepackage{makecell}
\newcolumntype{x}[1]{>{\centering\arraybackslash}p{#1}}
\usepackage{tikz}
\usepackage{pgfplots}
\usetikzlibrary{shapes.geometric, shapes.misc, positioning, arrows, arrows.meta, decorations.pathreplacing, decorations.pathmorphing, patterns, angles, quotes, calc}
\usepackage{booktabs}
\usepackage{xfrac}
\usepackage{siunitx}
\usepackage{centernot}
\usepackage{comment}
\usepackage{chngcntr}
\usepackage{caption}
\usepackage{subcaption}

\newtheorem{thm}{Theorem}
\newtheorem*{thm*}{Theorem}
\newtheorem{prop}[thm]{Proposition}
\newtheorem*{prop*}{Proposition}
\newtheorem{lemma}[thm]{Lemma}
\newtheorem*{lemma*}{Lemma}
\newtheorem{cor}[thm]{Corollary}
\newtheorem*{cor*}{Corollary}

\newtheorem*{cj*}{Conjecture}
\newtheorem{Def}[thm]{Definition}
\newtheorem*{Def*}{Definition}

\newtheorem*{question*}{Question}

\newtheorem*{problem*}{Problem}

\makeatletter
\def\thmhead@plain#1#2#3{%
  \thmname{#1}\thmnumber{\@ifnotempty{#1}{ }\@upn{#2}}%
  \thmnote{ {\the\thm@notefont#3}}}
\let\thmhead\thmhead@plain
\makeatother

\theoremstyle{definition}
\newtheorem{rem}[thm]{Remark}

\newcommand{\bb}{\begin{equation}\begin{aligned}\hspace{0pt}}
\newcommand{\bbb}{\begin{equation*}\begin{aligned}}
\newcommand{\ee}{\end{aligned}\end{equation}}
\newcommand{\eee}{\end{aligned}\end{equation*}}
\newcommand*{\coloneqq}{\mathrel{\vcenter{\baselineskip0.5ex \lineskiplimit0pt \hbox{\scriptsize.}\hbox{\scriptsize.}}} =}
\newcommand*{\eqqcolon}{= \mathrel{\vcenter{\baselineskip0.5ex \lineskiplimit0pt \hbox{\scriptsize.}\hbox{\scriptsize.}}}}

\newcommand{\eqt}[1]{\stackrel{\mathclap{\scriptsize \mbox{#1}}}{=}}
\newcommand{\leqt}[1]{\stackrel{\mathclap{\scriptsize \mbox{#1}}}{\leq}}

\newcommand{\geqt}[1]{\stackrel{\mathclap{\scriptsize \mbox{#1}}}{\geq}}
\newcommand{\ketbra}[1]{\ket{#1}\!\!\bra{#1}}

\newcommand{\sumno}{\sum\nolimits}

\renewcommand{\epsilon}{\varepsilon}

\newcommand{\id}{\mathds{1}}

\newcommand{\E}{\mathds{E}}

\DeclareMathOperator{\Tr}{Tr}

\DeclareMathAlphabet{\pazocal}{OMS}{zplm}{m}{n}

\DeclareMathOperator{\supp}{supp}
\DeclareMathOperator{\spec}{spec}

\newcommand{\EE}{\pazocal{E}}

\newcommand{\lsmatrix}{\left(\begin{smallmatrix}}
\newcommand{\rsmatrix}{\end{smallmatrix}\right)}

\stackMath

\stackMath

\makeatletter
\newcommand*\rel@kern[1]{\kern#1\dimexpr\macc@kerna}
\newcommand*\widebar[1]{%
  \begingroup
  \def\mathaccent##1##2{%
    \rel@kern{0.8}%
    \overline{\rel@kern{-0.8}\macc@nucleus\rel@kern{0.2}}%
    \rel@kern{-0.2}%
  }%
  \macc@depth\@ne
  \let\math@bgroup\@empty \let\math@egroup\macc@set@skewchar
  \mathsurround\z@ \frozen@everymath{\mathgroup\macc@group\relax}%
  \macc@set@skewchar\relax
  \let\mathaccentV\macc@nested@a
  \macc@nested@a\relax111{#1}%
  \endgroup
}

\counterwithin*{equation}{part}
\counterwithin*{thm}{part}
\counterwithin*{figure}{part}

\tikzset{meter/.append style={draw, inner sep=10, rectangle, font=\vphantom{A}, minimum width=30, line width=.8, path picture={\draw[black] ([shift={(.1,.3)}]path picture bounding box.south west) to[bend left=50] ([shift={(-.1,.3)}]path picture bounding box.south east);\draw[black,-latex] ([shift={(0,.1)}]path picture bounding box.south) -- ([shift={(.3,-.1)}]path picture bounding box.north);}}}
\tikzset{roundnode/.append style={circle, draw=black, fill=gray!20, thick, minimum size=10mm}}
\tikzset{squarenode/.style={rectangle, draw=black, fill=none, thick, minimum size=10mm}}

\definecolor{Blues5seq1}{RGB}{239,243,255}
\definecolor{Blues5seq2}{RGB}{189,215,231}
\definecolor{Blues5seq3}{RGB}{107,174,214}
\definecolor{Blues5seq4}{RGB}{49,130,189}
\definecolor{Blues5seq5}{RGB}{8,81,156}

\definecolor{Greens5seq1}{RGB}{237,248,233}
\definecolor{Greens5seq2}{RGB}{186,228,179}
\definecolor{Greens5seq3}{RGB}{116,196,118}
\definecolor{Greens5seq4}{RGB}{49,163,84}
\definecolor{Greens5seq5}{RGB}{0,109,44}

\definecolor{Reds5seq1}{RGB}{254,229,217}
\definecolor{Reds5seq2}{RGB}{252,174,145}
\definecolor{Reds5seq3}{RGB}{251,106,74}
\definecolor{Reds5seq4}{RGB}{222,45,38}
\definecolor{Reds5seq5}{RGB}{165,15,21}

\allowdisplaybreaks

\usepackage[most,breakable]{tcolorbox}
\newenvironment{boxedthm}[1]%
	{\expandafter\ifstrequal\expandafter{#1}{orange}{\begin{tcolorbox}[colback=red!15,colframe=orange!15,breakable,enhanced]}{\begin{tcolorbox}[colback=Blues5seq1,colframe=Blues5seq5,breakable,enhanced]}}%
	{\end{tcolorbox}}

\usepackage{authblk}

\renewcommand{\EE}[1]{\underset{\scaleobj{.8}{#1}}{\mathds{E}\,}}

\allowdisplaybreaks[1] 

\begin{document}

\title{Random purification channel made simple}
\author[1]{Filippo Girardi}
\author[1]{Francesco Anna Mele}
\author[1]{Ludovico Lami}

\affil[1]{Scuola Normale Superiore, Piazza dei Cavalieri 7, 56126 Pisa, Italy\vspace{-3em}}

\date{}
\setcounter{Maxaffil}{0}
\renewcommand\Affilfont{\itshape\small}

\maketitle
\begin{abstract}
The recently introduced random purification channel, which converts $n$ i.i.d.\ copies of any mixed quantum state into a uniform convex combination of $n$ i.i.d.\ copies of its purifications, has proved to be an extremely useful tool in quantum learning theory. Here we give a remarkably simple construction of this channel, making its known properties --- and several new ones --- immediately transparent. In particular, we show that the channel also purifies non-i.i.d.\ states: it transforms any permutationally symmetric state into a uniform convex combination of permutationally symmetric purifications, each differing only by a tensor-product unitary acting on the purifying system. We then apply the channel to give a one-line proof of (a stronger version of) the recently established Uhlmann's theorem for quantum divergences. 
\end{abstract}

\section{Introduction}
Recently, paralleling earlier findings in property testing~\cite{soleimanifar2022testingmatrixproductstates,chen2025localtestunitarilyinvariant}, Tang, Wright, and Zhandry~\cite{tang2025} established the following beautiful result: there exists a channel that transforms $n$ copies of an arbitrary mixed state $\rho$ into $n$ copies of a uniformly random purification of $\rho$. More precisely, they showed the following.

\begin{lemma}[(Random purification channel)]\label{lemma_purif}
    Let $\mathcal{H}_A$ be an Hilbert space. For any $n\geq 1$, there exists a quantum channel $\Lambda_{\rm purify}^{(n)}:\mathcal{L}(\mathcal{H}_A^{\otimes n})\to \mathcal{L}\big((\mathcal{H}_A\otimes\mathcal{H}_B)^{\otimes n}\big)$, where $\mathcal{H}_B$ is isomorphic to $\mathcal{H}_A$, such that, for any arbitrary $\rho_A\in\mathcal{D}(\mathcal{H}_A)$,
    \bb\label{eq:property}
    \Lambda^{(n)}_{\mathrm{purify}}(\rho_A^{\otimes n})
        = \EE{{U}_B}\!\left[
            (\id_A \otimes U_B)\,
            (\psi_\rho)_{AB}\,
            (\id_A \otimes U_B^\dagger)
        \right]^{\otimes n},
    \ee
    where the expectation is taken over Haar-random unitaries $U_B$ on $\mathcal{H}_B$, 
    $(\psi_\rho)_{AB}$ denotes any fixed purification of $\rho$ in $\mathcal{H}_A\otimes\mathcal{H}_B$, and $\id_A$ is the identity over $\mathcal{H}_A$.  
    In other words, this channel maps $n$ copies of $\rho_A$ to $n$ copies of a uniformly random purification of $\rho_A$.
\end{lemma}

Utsumi, Nakata, Wang, and Takagi \cite{utsumi2025} leveraged this technique to the fidelity estimation task between two given unknown states, providing an improvement upon the state-of-the-art sample complexity. Then, Pelecanos, Spilecki, Tang, and Wright~\cite{pelecanos2025} applied this lemma to obtain a remarkably simple proof of the sample complexity of quantum state tomography for mixed states, one of the central results in quantum learning theory. They also left the reader with an intriguing question: \emph{are there additional applications of the random purification channel, even beyond quantum learning theory?}

To this we could add another, related question. The construction of the random purification channel presented in~\cite{tang2025, pelecanos2025} is quite convoluted, as it makes use of some relatively heavy machinery from representation theory and Schur--Weyl duality. \emph{Is there a more transparent construction of the random purification channel?}

In this paper, we answer both questions in the affirmative. First, we present a simple proof of the above Lemma~\ref{lemma_purif}\,---\,substantially more elementary than the original proofs in~\cite{tang2025, pelecanos2025}. Second, we propose an application of the random purification channel in quantum Shannon theory~\cite{MARK}. Specifically, we use it to provide a remarkably simple proof of the recently established \emph{Uhlmann theorem for divergences}~\cite{Mazzola_2025}, while simultaneously deriving a strictly more general result.

\section{A simple construction of the random purification channel}

Let us start by fixing some terminology. Given a Hilbert space $\mathcal{H}$ and an operator $X_n$ on its $n$-fold tensor power $\mathcal{H}^{\otimes n}$, we say that $X_n$ is \emph{permutationally symmetric} if, for all permutations $\pi\in S_n$ of a set of $n$ elements, we have $P_\pi^{\vphantom{\dag}} X_n P_\pi^\dag = X_n$, where $P_\pi$ is the unitary that permutes the tensor factors of $\mathcal{H}^{\otimes n}$ according to $\pi$. Now, let $\Gamma_{AB} \coloneqq \ketbra{\Gamma}_{AB}$, with $\ket{\Gamma}_{AB} \coloneqq \sum_i \ket{i}_A\otimes \ket{i}_B$, be the un-normalised maximally entangled state on a bipartite system $AB$, where $\mathcal{H}_A\simeq \mathcal{H}_B$ are isomorphic and $\{\ket{i}_A\}_i$ and $\{\ket{i}_B\}_i$ are two local orthonormal bases. We define the following positive semi-definite operator on $\big(\mathcal{H}_A\otimes \mathcal{H}_B\big)^{\otimes n}$:
\bb\label{eq:R_n}
    R_n\coloneqq \EE{U_B}\!\left[
            (\id_{A^n} \otimes U_B^{\otimes n})\,
            \Gamma_{AB}^{\otimes n}\,
            (\id_{A^n} \otimes (U_B^\dagger)^{\otimes n})
        \right],
\ee
where $U_B$ is a random unitary on the $B$ system distributed according to the Haar measure.

\begin{lemma}\label{lem:commute} 
The operator $R_n$ defined in~\eqref{eq:R_n} commutes with every operator of the form 
$X_{A^n} \otimes \mathbb{1}_{B^n}$, where $X_{A^n}$ is a permutationally symmetric 
operator on $\mathcal{H}_A^{\otimes n}$. In particular, for any permutationally symmetric state $\rho_{A^n}$ on 
$\mathcal{H}_A^{\otimes n}$, we have
\bb
    \sqrt{\rho_{A^n}}\, R_n \,\sqrt{\rho_{A^n}}
    = 
    \sqrt{R_n}\, \rho_{A^n}\, \sqrt{R_n}\, .
\ee
\end{lemma}

\begin{proof} 
The key observation is that $R_n$ commutes with all local i.i.d.\ unitaries of the form $U_A^{\otimes n} \otimes \mathbb{1}_{B^n}$. Indeed,
\begin{align}
\left(U_A^{\otimes n} \otimes \mathbb{1}_{B^n}\right) R_n \left(U_A^{\otimes n} \otimes \mathbb{1}_{B^n}\right)^\dag &= \EE{W_B} \left[ \left((U_A \otimes W_B)\, \Gamma_{AB}\, (U_A \otimes W_B)^\dag \right)^{\otimes n} \right] \nonumber \\
&= \EE{W_B} \left[ \left(\big(\id_A \otimes W_B U_A^\intercal\big)\, \Gamma_{AB}\, \big(\id_A \otimes W_BU_A^\intercal\big)^\dag \right)^{\otimes n} \right] \\
&= \EE{W'_B} \left[ \left((\id_A \otimes W'_B)\, \Gamma_{AB}\, \big(\id_A \otimes W'_B)^\dag \right)^{\otimes n} \right] \nonumber \\
&= R_n\, , \nonumber 
\end{align}
where on the second line we applied the `transpose trick' and on the third we exploited the left- and right-invariance of the Haar measure.

Since $\left[R_n,\ U_A^{\otimes n} \otimes \mathbb{1}_{B^n}\right]=0$ for any unitary $U_A$, Lemma~\ref{lemma_commutant} implies that $\left[ R_n, \ X_{A^n} \otimes \mathbb{1}_{B^ n} \right] = 0$ for any permutationally symmetric 
operator $X_{A^n}$. As commuting operators are simultaneously diagonalisable, one sees that also $\left[\sqrt{R_n} ,\ X_{A^n} \otimes \mathbb{1}_{B^ n}\right] = 0$ for any permutationally symmetric operator $X_{A^n}$. But then
\bb
\sqrt{\rho_{A^n}} R_n \sqrt{\rho_{A^n}} = \sqrt{\rho_{A^n}}\sqrt{R_n} \sqrt{R_n} \sqrt{\rho_{A^n}} = \sqrt{R_n}\sqrt{\rho_{A^n}}\sqrt{\rho_{A^n}}\sqrt{R_n} = \sqrt{R_n}\rho_{A^n}\sqrt{R_n} ,
\ee
thereby completing the proof.

\end{proof}

\begin{boxedthm}{}
\begin{thm}[(A simple construction of the random purification channel)] Let $\mathcal{H}_A\simeq \mathcal{H}_B$ be two isomorphic Hilbert spaces. For any integer $n\geq 1$, let $R_n$ be the operator on $\big(\mathcal{H}_A\otimes \mathcal{H}_B\big)^{\otimes n}$ defined as in~\eqref{eq:R_n}. The channel 
\bb
    \Lambda^{(n)}(\,\cdot\,)\coloneqq \sqrt{R_n}\big(\,\cdot\,\otimes\,\id_{B^n}\big)\sqrt{R_n}
\ee
is exactly equal to the map $\Lambda^{(n)}_{\mathrm{purify}}$ of~\cite{tang2025,pelecanos2025}. In particular, it satisfies the property~\eqref{eq:property}, namely
\bb\label{eq:lambda_iid}
    \Lambda^{(n)}\big(\rho_A^{\otimes n}\big)
        = \EE{{U}_B}\!\left[
            (\id_A \otimes U_B)\,
            (\psi_\rho)_{AB}\,
            (\id_A \otimes U_B^\dagger)
        \right]^{\otimes n}
\ee
for all states $\rho_A$ on $\mathcal{H}_A$, where the expectation is taken over Haar-random unitaries $U_B$ on $\mathcal{H}_B$, $(\psi_\rho)_{AB}$ denotes any fixed purification of $\rho$ in $\mathcal{H}_A\otimes\mathcal{H}_B$, and $\id_A$ is the identity over $\mathcal{H}_A$. 

\smallskip
More generally, if $\rho_{A^n}$ is a permutationally symmetric state on $\mathcal{H}_A^{\otimes n}$, then $\Lambda^{(n)}$ maps it to a convex combination of permutationally symmetric purifications, each differing only by a tensor-product unitary acting on the purifying system:
    \bb
        \Lambda^{(n)}\big(\rho_{A^n}\big)
        = \EE{{U}_B}\!\left[
            (\id_A \otimes U_B^{\otimes n})\,
            (\psi_\rho)_{A^nB^n}\,\left(
            \id_A \otimes U_B^{\otimes n}\right)^{\dagger}
        \right]\,,
    \ee
where $(\psi_\rho)_{A^nB^n}$ is any fixed permutationally symmetric purification of $\rho_{A^n}$ in $\mathcal{H}_A^{\otimes n}\otimes\mathcal{H}_B^{\otimes n}$.  

\end{thm}
\end{boxedthm}

\begin{proof}
    The map $\Lambda^{(n)}$ is manifestly completely positive. It is also trace preserving, as, by the very definitions of $\Lambda^{(n)}$ and $R_n$, we have
    \bb
    \Tr\left[\Lambda^{(n)}(X_{A^n})\right] &\eqt{(i)} \Tr\big[\big(X_{A^n}\,\otimes\,\id_{B^n}\big)R_n\big]\\
    &=\Tr_{A^n}\!\left[ X_{A^n} \EE{{U}_B}\!\left[\Tr_{B^n}\!\big[
            (\id_{A^n} \otimes U_B^{\otimes n})\,
            \Gamma_{AB}^{\otimes n}\,
            (\id_{A^n} \otimes (U_B^\dagger)^{\otimes n})\big]
        \right]\right]\\
        &\eqt{(ii)}\Tr_{A^n}\!\big[ X_{A^n}\Tr_{B^n}\!\big[
            \Gamma_{AB}^{\otimes n}\big]\big] \\
        &\eqt{(iii)} \Tr_{A^n}\!\big[ X_{A^n}\big]\, ,
    \ee
    where in~(i) and~(ii) we have leveraged the cyclicity of the trace, and in~(iii) we have observed that $\Tr_B \Gamma_{AB} = \id_A$. Now, if $\rho_{A^n}$ is a permutationally symmetric state on 
$\mathcal{H}_A^{\otimes n}$, we immediately see that
    \bb
        \Lambda^{(n)}\big(\rho_{A^{n}}\big)&=\sqrt{R_n}\rho_{A^{n}}\sqrt{R_n}\\
        &\eqt{(iv)}\sqrt{\rho_{A^{n}}}R_n \sqrt{\rho_{A^{n}}} \\
        &=\EE{{U}_B}\!\left[
            (\id_{A^n} \otimes U_B^{\otimes n})\,
            \sqrt{\rho_{A^{n}}}\,\Gamma_{AB}^{\otimes n}\sqrt{\rho_{A^{n}}}\,
            (\id_{A^n} \otimes (U_B^\dagger)^{\otimes n})
        \right]\\
        &\eqt{(v)} \EE{{U}_B}\!\left[
            (\id_A \otimes U_B)^{\otimes n}\,
            (\psi_\rho^{\rm \,std})_{A^nB^n}\,
            (\id_A \otimes U_B^\dagger)^{\otimes n}
        \right],
    \ee
    where in~(iv) we have used Lemma~\ref{lem:commute}, and in~(v) we have noticed that 
    \bb
    (\psi_\rho^{\rm \,std})_{A^nB^n}\coloneqq \sqrt{\rho_{A^n}}\,\Gamma_{AB}^{\otimes n}\sqrt{\rho_{A^n}}
    \ee
    is the standard purification of $\rho_{A^n}$ on $\mathcal{H}_A^{\otimes n}\otimes \mathcal{H}_B^{\otimes n}$. In particular, if we take $\rho_{A^n}=\rho_A^{\otimes n}$, we immediately get~\eqref{eq:lambda_iid}.
    
    It remains to prove that $\Lambda^{(n)}$ is exactly equal to the map $\Lambda^{(n)}_{\mathrm{purify}}$ of~\cite{tang2025,pelecanos2025} (even for input states that are not of the form $\rho^{\otimes n}$). From the construction in~\cite{tang2025,pelecanos2025}, it easily follows that $\Lambda^{(n)}_{\mathrm{purify}}(U_\pi (\cdot )U_\pi^\dagger)=\Lambda^{(n)}_{\mathrm{purify}}(\cdot)$ for all $\pi\in S_n$. In addition, since $R_n \big((U_\pi)_{A_n}\otimes (U_\pi)_{B_n}\big) = R_n$, as $R_n$ is supported on the fully symmetric subspace of $\mathcal{H}_{AB}^{\otimes n}$, the same property is satisfied by our channel: $\Lambda^{(n)}\big(U_\pi^{\vphantom{\dag}} (\cdot ) U_\pi^\dag\big)=\Lambda^{(n)}(\cdot)$  for all $\pi\in S_n$. Consequently, both $\Lambda^{(n)}$ and $\Lambda^{(n)}_{\mathrm{purify}}$ are completely characterised by their action on the space $H_{d,n}^{\rm sym}$ of permutationally symmetric Hermitian operators on $\mathcal{H}_A^{\otimes n}$. Hence, by exploiting that (a) both $\Lambda^{(n)}$ and $\Lambda^{(n)}_{\mathrm{purify}}$ satisfy~\eqref{eq:property} and (b) the set $\{\rho^{\otimes n}: \rho\in\mathcal{D}(\mathcal{H}_A)\}$ spans $H_{d,n}^{\rm sym}$ (see Lemma~\ref{fn:ll} below), it follows that $\Lambda^{(n)}=\Lambda^{(n)}_{\mathrm{purify}}$.
\end{proof}

\begin{lemma}\label{fn:ll}
    The set $\{\rho^{\otimes n}: \rho\in\mathcal{D}(\mathcal{H}_A)\}$ spans the space of permutationally symmetric operators over $\mathcal{H}_A^{\otimes n}$.
\end{lemma}

\begin{proof}
It suffices to prove that any permutationally symmetric state $\rho_{A^n}$ is spanned by $\{\rho^{\otimes n}: \rho\in\mathcal{D}(\mathcal{H}_A)\}$. This can be seen by noticing that the standard purification $\Psi_{A^nB^n}$ of a permutationally symmetric state $\rho_{A^n}$ belongs to the symmetric subspace of $\mathcal{H}_A^{\otimes n}\otimes\mathcal{H}_B^{\otimes n}$. By~\cite[Eq.~(11b)]{harrow2013churchsymmetricsubspace}, we can write $\Psi_{A^nB^n}=\sum_i\lambda_i(\psi_i)_{AB}^{\otimes n}$, so that the claim follows by tracing out the system $B^n$.
\end{proof}

\begin{rem}[(Fixed-rank extension)]
The above construction reproduces the results of~\cite{tang2025,pelecanos2025} when no constrained is imposed on the rank. In Appendix~\ref{app:finite-rank-rpc}, instead, we show how to extend our construction to the setting where the rank of the input state is guaranteed to be at most $r$, where $r\leq d$ is known. In this case, we recover the result of~\cite{tang2025,pelecanos2025} that an $r$-dimensional purifying system $B$ suffices (Theorem~\ref{thm:finite-rank-rpc}).
\end{rem}

\section{Uhlmann's theorem for divergences (in one line)}

Throughout this section, we detail an application of the random purification channel to quantum Shannon theory. Namely, we use this channel to give simpler proofs of the main results of~\cite{Mazzola_2025, Fang2025-variational}, which we also extend to a broader class of quantum divergences.

\subsection{Divergences and weak quasi-concavity}

\begin{Def}[(Divergence)]
    A function $\mathbb{D}:\mathcal{D}(\mathcal{H})\times\mathcal{D}(\mathcal{H})\rightarrow \mathbb{R}\cup\{+\infty\}$ is said to be a \emph{divergence} if it satisfies the following property, known as \emph{data-processing inequality}: for every quantum channel $\Lambda$ and every pair of states $(\rho,\sigma)$, we have
\bb
    \mathbb{D}\big(\Lambda(\rho)\big\|\Lambda(\sigma)\big)\leq \mathbb{D}(\rho\|\sigma).
\ee
\end{Def}

The most famous example of a divergence is the \emph{Umegaki relative entropy}~\cite{Umegaki1962}, defined by
\begin{equation} \label{eq:umegaki}
    D(\rho\|\sigma)
    \coloneqq
    \begin{cases}
            +\infty &  \supp(\rho)\nsubseteq \supp(\sigma),\\
            \Tr\!\left[\rho\bigl(\log\rho - \log\sigma\bigr)\right] & \text{otherwise.}
        \end{cases}
\end{equation}
It is well known that the Umegaki relative entropy is \emph{weak concave}, meaning that, for any ensemble of states $\{(p_i,\rho_i)\}_i$, we have
\begin{equation}\label{eq_quasi_conc}
    D\!\left(\sum_i p_i \rho_i \,\middle\|\, \sigma\right)
    \ge \sum_i p_i\, D(\rho_i\|\sigma) + \sum_i p_i \log p_i\,.
\end{equation}
For our analysis of Uhlmann’s theorem for divergences, we will require only a weaker assumption, which we refer to as \emph{weak quasi-concavity}.

\begin{Def}[(Weak quasi-concavity)] \label{def:weak_conc}
    Let $\mathcal{H}$ be a Hilbert space and let $d$ be its dimension. We say that a divergence $\mathbb{D}(\,\cdot\,\|\,\cdot\,)$ satisfies the \emph{weak quasi-concavity} property if there exists a polynomial $P_d$ such that, for any $n\geq 1$, for any finite ensemble of states $\{(p_i,\rho_i)\}_{i=1,\dots, N}$ on $\mathcal{H}^{\otimes n}$ and any $\sigma\in\mathcal{D}(\mathcal{H}^{\otimes n})$, we have
\bb\label{eq:weak_conc}
    \mathbb{D}\left(\sum_{i=1}^N p_i\rho_i\,\middle\|\,\sigma\right)\geq \min_{1\leq i\leq N}\mathbb{D}(\rho_i\|\sigma)-\log P_d(N, s_\sigma),
\ee
    where $s_{\sigma}\coloneqq|{\rm spec}(\sigma)|$.
\end{Def}

\begin{rem}[(Quasi-concavity implies weak quasi-concavity)]\label{rem:8}
    The requirement~\eqref{eq:weak_conc} is weaker than the one in~\eqref{eq_quasi_conc}, as it immediately follows from the fact that $\sum_{i=1}^Np_i\mathbb{D}(\rho_i\|\sigma)\geq \displaystyle{\min_{1\leq i\leq N}}\mathbb{D}(\rho_i\|\sigma)$ and $\sum_{i=1}^Np_i\log p_i\ge -\log N$.
\end{rem}
Other than the Umegaki relative entropy, we can show that three other divergences satisfy weak quasi-concavity: the sandwiched R\'enyi divergences, the measured R\'enyi divergences, defined as follows, and the \# R\'enyi divergences~\cite{Fawzi_2021_defining}.

\begin{Def}[{(Sandwiched R\'enyi divergences~\cite{tomamichel12smooth_tutorial,newRenyi,Wilde2014})}]
    Let $\rho$ and $\sigma$ be states in $\mathcal{D}(\mathcal{H})$, and let $\alpha\in(0,1)\cup(1,\infty)$. Then we set
    \bb
    \tilde Q_\alpha(\rho\|\sigma) \coloneqq \Tr\left[\left(\sigma^{\frac{1-\alpha}{2\alpha}}\rho \sigma^{\frac{1-\alpha}{2\alpha}}\right)^{\alpha}\right] ,
    \ee
    where it is understood that $Q_\alpha(\rho\|\sigma) = +\infty$ if either $\alpha\in (1,\infty)$ and $\supp(\rho)\nsubseteq \supp(\sigma)$, or $\alpha\in (0,1)$ and $\supp(\rho) \perp \supp(\sigma)$. The \emph{sandwiched R\'enyi divergence} of order $\alpha\in (0,\infty]$ is then defined as
    \bb
    \tilde D_{\alpha}(\rho\|\sigma)\coloneqq\frac{1}{\alpha-1}\log \tilde Q_\alpha(\rho\|\sigma).
    \ee
    The limiting cases $\alpha=1$ and $\alpha=\infty$ are obtained by continuity: at $\alpha=1$ we have $\tilde{D}_1(\rho\|\sigma) \coloneqq \lim_{\alpha\to 1} \tilde{D}_\alpha(\rho\|\sigma) = D(\rho\|\sigma)$, where the Umegaki relative entropy is given by~\eqref{eq:umegaki}; at $\alpha=\infty$, instead, we have $\tilde D_\infty(\rho\|\sigma) \coloneqq \lim_{\alpha\to\infty} \tilde D_\alpha(\rho\|\sigma) = \log \left\|\sigma^{-1/2} \rho \sigma^{-1/2}\right\|_\infty = D_{\max}(\rho\|\sigma)$, also called the max-relative entropy~\cite{Datta08}. Again, it is understood that $D_{\max}(\rho\|\sigma) = +\infty$ if $\supp(\rho)\not\subseteq \supp(\sigma)$.
\end{Def}

\begin{Def}[{(Measured R\'enyi divergences~\cite{newRenyi,Wilde2014})}]
    Let $\rho$ and $\sigma$ be states in $\mathcal{D}(\mathcal{H})$, and let $\alpha\in(0,1)\cup(1,\infty)$. The \emph{measured R\'enyi divergence} of order $\alpha\in (0,\infty]$ can be formally defined as
    \bb
        D_{M,\alpha}(\rho\|\sigma)\coloneqq \sup_{\pazocal{M}} D_\alpha\big(\pazocal{M}(\rho)\big\|\pazocal{M}(\sigma)\big),
    \ee
    where $\pazocal{M}$ is any arbitrary measurement channel and $D_\alpha(P\|Q) \coloneqq \frac{1}{\alpha-1} \log \sum_x P(x)^\alpha Q(x)^{1-\alpha}$ is the classical R\'enyi divergence of order $\alpha$ (limiting cases are treated as before).
\end{Def}

\begin{prop}[(Weak quasi-concavity for many divergences)]\label{prop:weak} The following divergences are weakly quasi-concave:
\begin{itemize}
    \item Umegaki relative entropy $D$;
    \item sandwiched R\'enyi divergences $\tilde D_\alpha$ of order $\alpha\in(0,\infty)$;
    \item measured R\'enyi divergences $D_{M,\alpha}$ of order $\alpha\in(0,\infty)$;
    \item \# R\'enyi divergences $D^{\#}_{\alpha}$ of order $\alpha\in(1,\infty)$~\cite{Fawzi_2021_defining}.
\end{itemize}
\end{prop}

\begin{proof} See Appendix~\ref{ap:weak}.
\end{proof}

\begin{lemma}\label{lem:quasi_conc} A weakly quasi-concave divergence $\mathbb{D}$ satisfies
\bb
    \mathbb{D}\left(\mathbb{E}_{\rho\sim \nu}\rho^{\otimes n}\,\middle\|\,\sigma^{(n)}\right)\geq \min_{\rho\in\supp(\nu)}\mathbb{D}\big(\rho^{\otimes n}\big\|\sigma^{(n)}\big)-\log {\rm poly}_d(n)
\ee
for any arbitrary probability measure $\nu$ on $\mathcal{D}(\mathcal{H})$ and any permutationally symmetric state $\sigma^{(n)}\in\mathcal{D}(\mathcal{H})$, where \mbox{$d=\dim\mathcal{H}$}. 
\end{lemma}
\begin{proof}
    See Appendix~\ref{ap:quasi_conc}.
\end{proof}

Remarkably, any divergence $\mathbb{D}_\alpha$ with the following two properties turns out to be asymptotically equivalent to the sandwiched R\'enyi divergence of order $\alpha$:
\begin{itemize}
    \item $\mathbb{D}_\alpha$ reduces to the classical R\'enyi divergence of order $\alpha$ for commuting states,
    \item $\mathbb{D}_\alpha$ satisfies weak quasi-concavity.
\end{itemize}
This follows from~\cite[Proposition~4.12]{TOMAMICHEL} together with the observation that the pinching map $\pazocal{P}_\sigma$ can be implemented with a random unitary taken from an ensemble whose cardinality is exactly $|\spec(\sigma)|$ (see e.g.~\cite[Eq.~(2.59)]{TOMAMICHEL}).

\subsection{Uhlmann's theorem for divergences}

\begin{Def}
    Let $\sigma_A\in\mathcal{D}(\mathcal{H}_A)$ be a state and let $\mathcal{H}_B$ be a Hilbert space isomorphic to $\mathcal{H}_A$. Then, we define the set $\pazocal{C}_{AB}^{\sigma_A}$ of $B$\emph{-extensions} of $\sigma_A$ as
    \bb
        \pazocal{C}_{AB}^{\sigma_A}=\left\{\tilde \sigma_{AB}\in\mathcal{D}(\mathcal{H}_A\otimes\mathcal{H}_B)\,:\, \Tr_B\tilde\sigma_{AB}=\sigma_A\right\},
    \ee
    and the family $\mathcal{C}_{AB}^{\sigma_A}$ as the sequence $\mathcal{C}_{AB}^{\sigma_A}\coloneqq \left(\pazocal{C}_{A^nB^n}^{\sigma_A^{\otimes n}}\right)_{n\geq 1}$.\end{Def}
 According to standard conventions, we define
\bb
\mathbb{D}^\infty(\rho\|\sigma)&\coloneqq \liminf_{n\to\infty}\frac{1}{n}\mathbb{D}(\rho^{\otimes n}\|\sigma^{\otimes n}), \\
\mathbb{D}^\infty(\rho\|\mathcal{F})&\coloneqq \liminf_{n\to\infty}\frac{1}{n}\inf_{\sigma_n\in \pazocal{F}_n}\mathbb{D}(\rho^{\otimes n}\|\sigma_n),
\ee
where $\mathcal{F}$ is the sequence of families of states $\pazocal{F}_n\subseteq \mathcal{D}(\mathcal{H})$. Incidentally, since most useful quantum divergences are either subadditive or superadditive (or both), Fekete's lemma guarantees that for such divergences the above limit infimum is actually a standard limit.

\begin{boxedthm}{}
\begin{thm}[(Axiomatic Uhlmann's theorem)]\label{thm:Uhlmann}
Let $\mathbb{D}(\,\cdot\,\|\,\cdot\,)$ be a divergence that obeys weak quasi-concavity (Definition~\ref{def:weak_conc}). Then, given $\rho_A$ and $\sigma_A$ in $\mathcal{D}(\mathcal{H}_A)$, for any arbitrary extension $\rho_{AB}$ of $\rho_A$ we have
\bb\label{eq:Uhlmann}
    \mathbb{D}^{\infty}(\rho_A\|\sigma_A)= \mathbb{D}^\infty \left(\rho_{AB}\middle\|\, \mathcal{C}^{\sigma_A}_{AB} \right)\,.
\ee
Moreover, a sequence of optimisers $(\tilde \sigma_{A^nB^n})_n\in \mathcal{C}^{\sigma_A}_{AB}$ is given by
\bb\label{eq:almost_optimizers}
    \tilde \sigma_{A^nB^n}=\pazocal{E}^{\otimes n}\circ\Lambda_{\rm purify}^{(n)}\left(\sigma_A^{\otimes n}\right)\,,
\ee
where $\pazocal{E}$ is any channel that, by acting only on the purifying system, maps a fixed purification of $\rho_A$ to the extension $\rho_{AB}$.
\end{thm}
\end{boxedthm}

A concise proof of the existence of the channel $\pazocal{E}$ for any extension $\rho_{AB}$ of $\rho_A$ can be found in \cite[Section~III]{squashed}. Essentially, it relies on the fact that a canonical purification of $\rho_A$ can always be transformed into a purification $\ket{\Psi}_{ABE}$ of $\rho_{AB}$ via a local isometry on the purifying system. Tracing away $E$ gives a Strinespring representation of the channel $\pazocal{E}$.

A remarkable property of the sequence of optimisers~\eqref{eq:almost_optimizers} is its \emph{universality}, i.e.\ the fact that it does not depend on the specific divergence $\mathbb{D}$. We also observe that, whenever $\mathbb{D}$ is additive, \eqref{eq:Uhlmann} simply becomes
\bb
    \mathbb{D}(\rho_A\|\sigma_A)= \mathbb{D}^\infty \left(\rho_{AB}\middle\|\, \pazocal{C}^{\sigma_A}_{AB} \right).
\ee

\begin{proof}[Proof of Theorem~\ref{thm:Uhlmann}.]
The right-hand-side of~\eqref{eq:lambda_iid} does not depend on the specific choice of the purification $\psi_\rho$. In fact, 
any other purification $\psi'_\rho$ can be obtained via a unitary map $V_B$ acting on $\psi_\rho$, i.e.\ $\psi'_\rho=V_B\psi_\rho V_B^\dagger$, and the Haar measure is right invariant. For this reason, we are going to write~\eqref{eq:lambda_iid} in the compact form
\bb \label{eq:property2}
\Lambda_{\rm purify}^{(n)}\big(\rho_A^{\otimes n}\big)= \EE{\ket{\psi_\rho}} \psi_{\!\rho}^{\otimes n}.
\ee
The inequality $\mathbb{D}^{\infty}(\rho_A\|\sigma_A)\leq \mathbb{D}^\infty \left(\rho_{AB}\middle\|\, \mathcal{C}^{\sigma_A}_{AB} \right)$ immediately follows from the data-processing inequality for $\mathbb{D}$, by applying the channel ${\rm Id}_{A^n}\otimes\Tr_{B^n}[\,\cdot\,]$ in the very definition of the right-hand-side of~\eqref{eq:Uhlmann} for any $n\geq 1$. 

Let us now focus on the converse inequality. First, it suffices to prove it for the case where $\rho_{AB} = \psi_{AB}$ is a purification of $\rho_A$. Indeed, any other extension $\rho_{AB'}$ can be obtained by applying a suitable quantum channel to the purifying system, say, $\rho_{AB'} = \pazocal{E}_{B\to B'}(\psi_{AB})$. Hence, 
\bb
\mathbb{D} \left(\rho_{AB'}^{\otimes n} \middle\|\, \pazocal{C}^{\sigma_A^{\otimes n}}_{A^n{B'}^n} \right) \leq \inf_{\tilde\sigma \in \pazocal{C}^\sigma_n} \mathbb{D} \left(\rho_{AB'}^{\otimes n} \middle\|\, \pazocal{E}_{B\to B'}^{\otimes n}(\tilde{\sigma}_{A^nB^n}) \right) \leq \inf_{\tilde\sigma \in \pazocal{C}^\sigma_n} \mathbb{D} \left(\psi_{AB}^{\otimes n}\, \middle\|\, \tilde{\sigma}_{A^n B^n} \right) = \mathbb{D} \left(\psi_{AB}^{\otimes n}\, \middle\|\, \pazocal{C}^\sigma_n \right) .
\ee
Here, the first inequality holds by taking as ansatzes all extensions of $\sigma_A^{\otimes n}$ to $A^n{B'}^n$ that are of the form $\pazocal{E}_{B\to B'}^{\otimes n}(\tilde\sigma_{A^nB^n})$, where $\tilde\sigma_{A^nB^n} \in \pazocal{C}^\sigma_n \coloneqq \pazocal{C}^{\sigma_A^{\otimes n}}_{A^nB^n}$; the second inequality, instead, is simply data-processing. Now, if we could show that the right-hand-side of the above equation is upper bounded by $n\mathbb{D}^\infty(\rho_A\|\sigma_A)$ up to terms that are of order $o(n)$, the proof would be complete. To this end, we write
\bb\label{eq:inequalities}
    \frac{1}{n}\mathbb{D}\left(\rho_A^{\otimes n}\middle\|\sigma_A^{\otimes n}\right)
    & \geqt{(i)}\frac{1}{n}\mathbb{D}\left(\EE{\ket{\psi_\rho}}\psi_{\!\rho}^{\otimes n}\,\middle\|\, \EE{\ket{\phi_\sigma}}\phi_{\sigma}^{\otimes n}\right)\\[4pt]
    &\geqt{(ii)} \frac{1}{n}\min_{\ket{\psi_\rho}}\mathbb{D}\left(\psi_{\!\rho}^{\otimes n}\, \middle\|\, \EE{\ket{\phi_\sigma}}\phi_{\sigma}^{\otimes n}\right)-\tfrac{\log{\rm poly }(n)}{n}\\[4pt]
    &\eqt{(iii)} \frac{1}{n}\mathbb{D}\left(\bar \psi^{\otimes n}\,\middle\|\, \EE{\ket{\phi_\sigma}}\phi_{\sigma}^{\otimes n}\right)-\tfrac{\log{\rm poly }(n)}{n} \\[4pt]
    &\geqt{(iv)} \frac{1}{n}\inf_{\tilde \sigma\in \pazocal{C}_n^{\sigma}}\mathbb{D}\left(\bar \psi_{AB}^{\otimes n}\,\middle\|\, \tilde \sigma_{A^nB^n} \right)-\tfrac{\log{\rm poly }(n)}{n}\, ,
\ee
where $\pazocal{C}_n^{\sigma} = \pazocal{C}_{A^nB^n}^{\sigma_A^{\otimes n}}$, as before; in~(i) we have used the data-processing inequality for $\mathbb{D}$ with the universal purifying map $\Lambda_{\rm purify}^{(n)}$, which acts on i.i.d.\ states as in~\eqref{eq:property}; in~(ii) we have leveraged the weak quasi-concavity of $\mathbb{D}$ to apply Lemma~\ref{lem:quasi_conc}; in~(iii) we have remarked that, for any fixed purification $\psi_\rho$, we can apply a local unitary on the system $B$, implemented by a unitary channel $\pazocal{U}_B$, to get $\bar \psi_{AB}$ out of $\psi_\rho$, namely $\bar \psi_{AB} = \pazocal{U}_B(\psi_\rho)$; in particular, by the unitary invariance of $\mathbb{D}$ --- which follows from the data-processing inequality --- and by the left-invariance of the Haar measure, we have
\bb
    \mathbb{D}\left(\psi_{\!\rho}^{\otimes n}\middle\| \;\EE{\ket{\phi_\sigma}}\phi_{\sigma}^{\otimes n}\right)=\mathbb{D}\left(\big(\pazocal{U}_{B}(\psi_\rho)\big)^{\otimes n}\middle\| \,\EE{\ket{\phi_\sigma}}\big(\pazocal{U}_B(\phi_{\sigma})\big)^{\otimes n}\right) = \mathbb{D}\left(\bar \psi_{AB}^{\otimes n}\middle\|\, \EE{\ket{\phi_\sigma}}\phi_{\sigma}^{\otimes n}\right);
\ee
finally, in (iv) we have noticed that $\EE{\ket{\phi_\sigma}}\phi_{\sigma}^{\otimes n}\in \pazocal{C}_n^{\sigma}$.
Therefore, taking the limit $n\to\infty$ in~\eqref{eq:inequalities}, we get
\bb
    \mathbb{D}^{\infty}(\rho_A\|\sigma_A) \geq \lim_{n\to \infty}\frac{1}{n}\inf_{\tilde \sigma\in \pazocal{C}_{n}^{\sigma}}\mathbb{D}\left(\psi_{AB}^{\otimes n}\middle\| \tilde \sigma_{A^nB^n} \right) = \mathbb{D}^\infty \left(\psi_{AB}\middle\|\, \mathcal{C}^{\sigma_A}_{AB}  \right).
\ee
In particular, this proof immediately implies that the sequence $(\tilde \sigma_{A^nB^n})_n\in \mathcal{C}_{AB}^{\sigma_A} $ given by~\eqref{eq:almost_optimizers} achieves the right-hand side of~\eqref{eq:Uhlmann}.
\end{proof}

\begin{rem} In order to emphasise how powerful the map $\Lambda_{\rm purify}^{(n)}$ is, we show that the previous proof can be compactified to a single line\footnote{
A fussy reader might complain that any proof, with a sufficiently small font, can fit on one line. We simply reply that our line is actually legible.}:
\begin{equation*}
\resizebox{.995\hsize}{!}{%
    $\frac1n \mathbb{D}\big(\bar \psi^{\otimes n}\big\|\,\pazocal{C}^\sigma_{n}\big) \geq \frac1n \mathbb{D}\big(\rho_A^{\otimes n}\big\|\sigma_A^{\otimes n}\big)
    \geq\frac{1}{n}\mathbb{D}\Big(\EE{\ket{\psi_\rho}}\psi_{\!\rho}^{\otimes n}\Big\|\, \EE{\ket{\phi_\sigma}}\phi_{\sigma}^{\otimes n}\Big) \gtrsim \frac{1}{n}\min\limits_{\ket{\psi_\rho}}\mathbb{D}\Big(\psi_{\!\rho}^{\otimes n}\Big\|\, \EE{\ket{\phi_\sigma}}\phi_{\sigma}^{\otimes n}\Big) = \frac{1}{n}\mathbb{D}\Big(\bar \psi^{\otimes n}\Big\|\, \EE{\ket{\phi_\sigma}}\phi_{\sigma}^{\otimes n}\Big) \geq \frac{1}{n}\mathbb{D}\big(\bar \psi^{\otimes n}\big\|\, \pazocal{C}^\sigma_n \big),$
}
\end{equation*}
where $\pazocal{C}_n^{\sigma}\coloneqq \pazocal{C}_{A^nB^n}^{\sigma_A^{\otimes n}}$, and the inequality $\gtrsim$ holds up to terms that vanish as $n\to \infty$.
\end{rem}

The following result was previously found in~\cite[Corollary~5]{Fang2025-variational}. Here, we show that it can also be seen as a simple consequence of the above general result. 

\begin{cor}[{\cite[Corollary~5]{Fang2025-variational}}] \label{cor:Umegaki} Given $\alpha\in [1/2,\infty]$, let $\tilde D_\alpha$ be the sandwiched relative entropy of order $\alpha$. Then, given $\rho_A$ and $\sigma_A$ in $\mathcal{D}(\mathcal{H}_A)$, for any arbitrary extension $\rho_{AB}$ of $\rho_A$ we have
\bb\label{eq:Uhlmann_sandwiched}
    \tilde D_\alpha(\rho_A\|\sigma_A)= \tilde D_\alpha^\infty \left(\rho_{AB}\middle\|\, \mathcal{C}^{\sigma_A}_{AB} \right)=\lim_{n\to\infty}\frac 1 n \tilde D_\alpha\left(\rho_{AB}^{\otimes n}\,\middle\|\, \pazocal{E}^{\otimes n}\circ\Lambda_{\rm purify}^{(n)}\left(\sigma_A^{\otimes n}\right)\right),
\ee 
where $\pazocal{E}$ is any channel that, by acting only on the purifying system, maps a fixed purification of $\rho_A$ to the extension $\rho_{AB}$.
In particular, for $\alpha=1$, the identity~\eqref{eq:Uhlmann_sandwiched} holds for the Umegaki relative entropy $D$.
\end{cor}

\begin{proof}
It is well known that $\tilde D_{\alpha}$ satisfies the data-processing inequality if $1/2\leq \alpha\leq \infty$. Weak quasi-concavity was proved in Proposition~\ref{prop:weak}. Therefore, due to Theorem~\ref{eq:Uhlmann} and to the additivity of $\tilde D_{\alpha}$, the proof is complete.
\end{proof}



\subsection*{Acknowledgements}
We thank Giulia Mazzola, Lukas Schmitt, David Sutter, Lennart Bittel, and Antonio Anna Mele for valuable comments on an earlier version of the manuscript. We acknowledge financial support from the European Union under the European Research Council (ERC Grant Agreement No.~101165230).

\bibliography{biblio}

\appendix

\section{Some proofs}

\subsection{Proof of Proposition~\ref{prop:weak}}\label{ap:weak}

We need a couple of preliminary lemmas.

\begin{lemma}[{\cite[Lemma~4.11]{TOMAMICHEL}}] \label{lem:marco} Let $\tilde D_{\alpha}$ be the sandwiched R\'enyi divergence of order $\alpha\in (0,\infty)$, let $\rho$ and $\sigma$ be states in $\mathcal{D}(\mathcal{H})$, and let $\pazocal{P}_\sigma$ be the pinching map on $\sigma$. Then we have
\bb
    \tilde D_{\alpha}(\rho\|\sigma)\leq \tilde D_{\alpha}(\pazocal{P}_\sigma(\rho)\|\sigma)+\eta_\alpha \log|{\rm spec }(\sigma)|\, , \qquad \eta_\alpha\coloneqq
    \begin{cases}
        1 & 0<\alpha\leq 2\, ,\\
        \frac{\alpha}{\alpha-1} & \alpha>2\, .
    \end{cases}
\ee
\end{lemma}

Now we have all the ingredients to prove Proposition~\ref{prop:weak}.

\begin{proof}[Proof of Proposition~\ref{prop:weak}]
     Let $\rho\coloneqq\sum_{i=1}^N p_i\rho_i$. We give a short proof for each case.\\[0.5em]
    \noindent\emph{Umegaki relative entropy $D$ (i.e. case $\alpha=1$).} The result immediately follows from~\eqref{eq_quasi_conc} combined with Remark~\ref{rem:8}. \\[0.5em]
    \noindent\emph{Sandwiched R\'enyi divergence $\tilde D_\alpha$ of order $\alpha\in(0,1)$}. We have
    \bb\label{eq:q_alpha}
         \tilde Q_\alpha(\rho\|\sigma)&\leqt{(i)} \sum_{i=1}^{N}p_i^\alpha\tilde Q_\alpha(\rho_i\|\sigma)
         \leq \left(\sum_{i=1}^{N}p_i^\alpha\right)\max_{1\leq j\leq N}\tilde Q_\alpha(\rho_i\|\sigma)
         \leqt{(ii)} N^{1-\alpha} \max_{1\leq i\leq N}\tilde Q_\alpha(\rho_i\|\sigma) ,
    \ee
    where in (i) we have used~\cite[Proposition~III.8]{Mosonyi_2015_coding}, and in (ii) we have leveraged the concavity of $x\mapsto x^\alpha$, namely
    \bb
        \sum_{i=1}^{N}p_i^\alpha=N\sum_{i=1}^{N}\frac{1}{N}p_i^\alpha\leq N\left(\sum_{i=1}^N\frac{1}{N}p_i\right)^\alpha = N^{1-\alpha}.
    \ee
    From~\eqref{eq:q_alpha} we immediately get
    \bb
        \tilde D_\alpha (\rho\|\sigma)\geq \min_{1\leq i \leq N}\tilde D_\alpha (\rho_i\|\sigma)-\log N.
    \ee
    \noindent\emph{Sandwiched R\'enyi divergence $\tilde D_\alpha$ of order $\alpha\in(1,\infty)$.} We notice that $p_k\rho_k\leq \rho$ for any $1\leq k\leq N$. Therefore,
    \bb\label{eq:minorisation}
        \tilde Q_\alpha(\rho\|\sigma)=\Tr\left[(\sigma^{\frac{1-\alpha}{2\alpha}}\rho\sigma^{\frac{1-\alpha}{2\alpha}})^\alpha\right]\geq p_k^\alpha\Tr\left[(\sigma^{\frac{1-\alpha}{2\alpha}}\rho_k\sigma^{\frac{1-\alpha}{2\alpha}})^\alpha\right]\,,
    \ee
    where we observed that $\Tr [A^\alpha] \geq \Tr[B^\alpha]$ if $A \geq B \geq 0$.
    In particular, since $\sum_{i=1}^Np_i=1$, there is at least one $\bar k$ such that $p_{\bar k}\geq 1/N$. Hence,
    \bb
        \tilde D_\alpha(\rho\|\sigma)&\geq  \frac{1}{\alpha-1}\Tr\left[(\sigma^{\frac{1-\alpha}{2\alpha}}\rho_{\bar k}\sigma^{\frac{1-\alpha}{2\alpha}})^\alpha\right] +\frac{\alpha}{\alpha-1}\log p_{\bar k}\\
        &\geq \tilde D_\alpha(\rho_{\bar k}\|\sigma)-\frac{\alpha}{\alpha-1}\log N\\
        &\geq \min_{1\leq i \leq N}\tilde D_\alpha(\rho_i\|\sigma)-\frac{\alpha}{\alpha-1}\log N.
    \ee 
    \emph{Measured R\'enyi divergences $D_{M,\alpha}$ of order $\alpha\in(0,\infty)$.}
    We have
    \vspace{-4pt}
    \bb
        D_{M,\alpha}(\rho\|\sigma) &\geqt{(iii)} D_{M,\alpha}\big(\pazocal{P}_\sigma(\rho)\big\|\sigma\big) \\
        &\eqt{(iv)} \tilde D_{\alpha}\big(\pazocal{P}_\sigma(\rho)\big\|\sigma\big)\\
        &\geqt{(v)} \tilde D_{\alpha}(\rho\|\sigma)-\eta_\alpha \log|{\rm spec }(\sigma)|\\
        &\geqt{(vi)} \min_{1\leq i\leq n}\tilde D_{\alpha}(\rho_i\|\sigma)-\log {\rm poly}_d(N, s_\sigma)-\eta_\alpha \log s_\sigma\\
        &\geqt{(vii)} \min_{1\leq i\leq n}\tilde D_{M,\alpha}(\rho_i\|\sigma)-\log {\rm poly}_d(N, s_\sigma)-\eta_\alpha \log s_\sigma
    \ee
    where: (iii)~holds by data-processing; in~(iv) we have noticed that $\pazocal{P}_\sigma(\rho)$ and $\sigma$ commute, so that the measured and the sandwiched R\'enyi divergences coincide; in~(v) we have leveraged Lemma~\ref{lem:marco}; in~(vi) we have recalled that $\tilde D_{\alpha}$ is weakly quasi-concave, setting also $s_\sigma\coloneqq|{\rm spec }(\sigma)|$; in~(vii) we have used the fact that $D_{\alpha}\geq D_{M,\alpha}$, which is an immediate consequence of the data-processing inequality. \\[0.5em]
    \emph{\# R\'enyi divergences $D^{\#}_{\alpha}$ of order $\alpha\in(1,\infty)$.} Leveraging~\cite[Proposition~3.4]{Fawzi_2021_defining}, we have
    \bb
        D^{\#}_{\alpha}(\rho\|\sigma)&\geq D_{M,\alpha}(\rho\|\sigma)\\
        &\geqt{(viii)} \min_{1\leq i\leq N}D_{M,\alpha}(\rho_i\|\sigma)-\log{\rm poly}_d(N,s_\sigma)\\
        &\geq  \min_{1\leq i\leq N}D^{\#}_{\alpha}(\rho_i\|\sigma)-\log{\rm poly}_d(N,s_\sigma)-\frac{\alpha}{\alpha-1}\log|{\rm spec}(\sigma)|,
    \ee    
    where in~(viii) we have used the weak quasi-concavity of $D_{M,\alpha}$.
    \end{proof}

\subsection{Proof of Lemma~\ref{lem:quasi_conc}}\label{ap:quasi_conc}

\begin{proof}
    The real vector space $H_{d,n}^{\rm sym}$ of permutationally symmetric Hermitian operators on $\mathcal{H}^{\otimes n}\simeq \big(\mathbb{C}^d\big)^{\otimes n}$, by Schur-Weyl duality, has the form
    \bb
        H_{d,n}^{\rm sym}=\bigoplus_{\lambda\in\pazocal{Y}_n^d}\pazocal{U}_\lambda\otimes\id_{\pazocal{V}_\lambda},
    \ee
    where $\lambda$ is an index ranging on the set $\pazocal{Y}_d^n$ of Young diagrams with size $n$ and depth at most $d$, and $\pazocal{U}_\lambda$ and $\pazocal{V}_\lambda$ are irreps of the special unitary group ${\rm SU}(d)$ and of the symmetric group $S_n$, respectively. 
    Since $\dim \pazocal{U}_\lambda\leq (n+1)^{d(d-1)/2}$ and $|\pazocal{Y}_n^d|\leq (n+1)^{d-1}$, we can upper bound $\dim H_{d,n}^{\rm sym}\leq (n+1)^{(d-1)\left(\frac{d}{2}+1\right)}$.
    By Carath\'eodory's theorem, since $\mathbb{E}_{\rho\sim \nu}\rho^{\otimes n}\in H_{d,n}^{\rm sym}$, we can write it as a convex combination of at most $N=(n+1)^{(d-1)\left(\frac{d}{2}+1\right)}+1$ terms of the form $\rho^{\otimes n}$, where $\rho\in\supp(\nu)$:
    \bb
        \mathbb{E}_{\rho\sim \nu}\rho^{\otimes n}=\sum_{i=1}^Np_i\rho_i^{\otimes n}\qquad \rho_i\in\supp(\nu).
    \ee
     The cardinality of the spectrum of $\sigma^{(n)}$ is polynomial in $n$ due to the permutational invariance of $\sigma^{(n)}$. Indeed, if $\ket{\psi_\gamma}$ is an eigenvector with eigenvalue $\gamma$, then also $\psi'_\gamma \coloneqq\frac{1}{n!}\sum_{\pi\in S_n}U_\pi \psi_\gamma$, which belongs to the symmetric subspace of $\mathcal{H}^{\otimes n}$. Such space has dimension $\binom{n+d-1}{n}\leq (n+d-1)^{d-1}={\rm poly}_d(n)$. Hence, $s_\sigma=|{\rm spec}(\sigma)|$ grows at most polynomially in $n$. Therefore, in the definition of weak quasi-concavity, we can upper bound $P_d(N, s_\sigma)$ --- for suitable $a,b,c>0$ possibly depending on $d$ but not on $n$ --- as
     \bb
        P_d(N, s_\sigma)\leq a(N+s_\sigma)^b+c\leqt{(i)} {\rm poly}_d(n),
     \ee
     where in (i) we have used the polynomial upper bounds on $N$ and $s_\sigma$. This concludes the proof.
\end{proof}

\subsection{The commutants of the $\boldsymbol{n}$-fold tensor powers of the state and the unitary representation coincide}

\begin{lemma}\label{lemma_commutant}
    Let $\Theta_n$ be a linear operator over $\mathcal{H}^{\otimes n}$. Then the following four statements are equivalent:
    \begin{enumerate}[(a)]
        \item $[\Theta_n,U^{\otimes n}]=0$ for all unitaries $U$ on $\mathcal{H}$;
        \item $[\Theta_n,\rho^{\otimes n}]=0$ for all states $\rho$ on $\mathcal{H}$;
        \item $[\Theta_n,X_n]=0$ for all permutationally symmetric operators $X$ on $\mathcal{H}^{\otimes n}$.
        \item $\big[\Theta_n, \sum_{j=1}^nH^{(j)}\big]=0$ for all Hermitian operators $H$, where we denoted as $H^{(j)}\coloneqq \mathbb{1}^{\otimes (j-1)}\otimes H\otimes \mathbb{1}^{\otimes (n-j)}$ the operator that acts as $H$ on the $j$th system and as the identity on all the other systems.
    \end{enumerate}
\end{lemma}

\begin{proof}
Statements~(b) and~(c) are equivalent by Lemma~\ref{fn:ll}. Moreover, (c)~implies~(a), due to the fact that $U^{\otimes n}$ is manifestly permutationally symmetric.


We now show that~(a) implies~(d). If $\big[\Theta_n,U^{\otimes n}\big]=0$ for every unitary $U$, then in particular $\big[\Theta_n,(e^{itH})^{\otimes n}\big]=0$ for every Hermitian operator $H$ and every $t\in\mathbb{R}$. Differentiating with respect to $t$ at $t=0$ yields $\big[\Theta_n, \sum_{j=1}^n{H}^{(j)}\big]=0$ for all Hermitian operators $H$, which is exactly~(d). (The reasoning can be also reversed to show that~(a) and~(d) are in fact equivalent.)

To conclude the proof, it suffices to show that~(d) implies~(b). If $\big[\Theta_n,\sum_{j=1}^n H^{(j)}\big]=0$ for all Hermitian operators $H$, then 
\bb
0 = \Big[\Theta_n,\, \exp\Big[ \sumno_{j=1}^n H^{(j)} \Big] \Big] = \Big[\Theta_n,\big(e^{H}\big)^{\otimes n}\Big]
\ee
for every Hermitian operator $H$. Since any strictly positive state $\rho$ can be written as $\rho=\frac{e^{H}}{\Tr e^{H}}$ for some Hermitian operator $H$, we obtain $\big[\Theta_n,\rho^{\otimes n}\big]=0$ for all strictly positive states $\rho$. By continuity, it follows that $\big[\Theta_n, \rho^{\otimes n}\big]=0$ for all states $\rho$, which is exactly~(b).
\end{proof}

\section{Finite-rank random purification channel}
\label{app:finite-rank-rpc}

In the main text, the purifying system has the same dimension as the input
system. In this appendix we record a finite-rank variant of the same
construction. The point is that, if the input state has rank at most \(r\), then
an \(r\)-dimensional purifying system is enough.

Throughout this appendix, let
\bb
    \mathcal H_A\simeq \mathbb C^d,
    \qquad
    \mathcal H_B\simeq \mathbb C^r,
    \qquad
    r\le d .
\ee
For every isometry \(V:\mathcal H_B\to\mathcal H_A\), define the unnormalised
maximally entangled vector
\bb
    \ket{\Gamma_V}_{AB}
    \coloneqq 
    (V\otimes I_B)\ket{\Gamma_r}\, ,
    \qquad
    \ket{\Gamma_r} \coloneqq \sum_{i=1}^r \ket{i}\otimes \ket{i}\, .
\ee
Thus $\ket{\Gamma_V}$ is maximally entangled between $B$ and
the $r$-dimensional subspace $V\mathcal{H}_B\subseteq \mathcal{H}_A$.

Fix an isometry $V_0:\mathcal{H}_B\to\mathcal{H}_A$. Let $U_A$ be Haar-random on $\mathcal{H}_A$, and set $V=U_A V_0$. We write
$\mathbb E_V$ for expectation over this random isometry. By Haar invariance,
this distribution is independent of the choice of $V_0$. Define
\bb
R_{r,n} \coloneqq \E_V \left[\ketbra{\Gamma_V}^{\otimes n}\right] .
\ee
This is the finite-rank analogue of the random maximally entangled operator from
the main text. The only difference is that, for \(r<d\), its partial trace over
the purifying system is not the identity. Instead,
\bb
\Tr_{B^n} R_{r,n} = \E_V \left[ (VV^\dagger)^{\otimes n} \right] \eqqcolon M_{r,n}\, .
\ee
We define
\bb
\widehat{R}_{r,n} \coloneqq  \left(M_{r,n}^{-1/2}\otimes I_{B^n}\right) R_{r,n}\left(M_{r,n}^{-1/2}\otimes I_{B^n}\right),
\ee
where the inverse is taken on the support of the positive semidefinite operator $M_{r,n}$. Let $\Pi_{r,n}$ be the orthogonal projection onto $\supp(M_{r,n})$. Then
\bb
\Tr_{B^n}\widehat{R}_{r,n} = \Pi_{r,n}\, ,
\ee
because 
\bb
\Tr_{B^n} \widehat{R}_{r,n} &= M_{r,n}^{-1/2} \Tr_{B^n}[R_{r,n}] M_{r,n}^{-1/2} \\
&= M_{r,n}^{-1/2} M_{r,n} M_{r,n}^{-1/2} \\
&= \Pi_{r,n}.
\ee
We now define a map
\bb
\Lambda_r^{(n)}(X) \coloneqq  \sqrt{\widehat{R}_{r,n}} (X\otimes I_{B^n}) \sqrt{\widehat{R}_{r,n}} + \Tr\big[(I_{A^n}-\Pi_{r,n})X\big]\,\tau,
\ee
where $\tau = \tau_{A^nB^n}$ is any fixed state on $A^nB^n$. The second term only makes the map trace-preserving on all inputs. It will vanish on the rank-constrained inputs considered below.

\begin{lemma}
\label{lem:rank-r-channel}
The map $\Lambda_r^{(n)}$ is a quantum channel.
\end{lemma}

\begin{proof}
Complete positivity is immediate from the definition, and, as mentioned, the addition of the second term also makes it trace preserving.
\end{proof}

The next lemma is the only new ingredient. It says that, once a rank-$r$ support is fixed, the operator $\widehat{R}_{r,n}$ becomes the full-rank random maximally entangled operator on that support.

\begin{lemma}[(Compression to a fixed support)] \label{lem:rank-r-compression}
Let $S\subseteq\mathcal{H}_A$ be an $r$-dimensional subspace, and let $P_S$
be the projection onto $S$. Fix an isometry $W:\mathcal{H}_B \to \mathcal{H}_A$ whose image is $S$. Then
\bb
\big(P_S^{\otimes n}\otimes I_{B^n}\big)\, \widehat{R}_{r,n}\,\big(P_S^{\otimes n}\otimes I_{B^n}\big) = \E_{U\in\mathsf U(r)} \left[ \ketbra{\Gamma_{WU}}^{\otimes n} \right].
\ee
\end{lemma}

\begin{proof}
Let $V:\mathcal H_B\to\mathcal H_A$ be the random isometry used in the
definition of $R_{r,n}$, and set
\bb
C\coloneqq W^\dagger V\, .
\ee
Here and below, $\mathbb E_C$ means the expectation induced by first sampling $V$ and then setting $C=W^\dagger V$.

Now, since $P_S=WW^\dagger$, we have
\bb
(P_S\otimes I_B)\ket{\Gamma_V} = \ket{\Gamma_{WC}} = (W\otimes I_B)\ket{\Gamma_C}\, .
\ee
Therefore,
\bb
\big(P_S^{\otimes n}\otimes I_{B^n}\big)\, R_{r,n}\, \big(P_S^{\otimes n}\otimes I_{B^n}\big) = \big(W^{\otimes n}\otimes I_{B^n}\big)\, \E_C \left[ \ketbra{\Gamma_C}^{\otimes n} \right] \big((W^\dagger)^{\otimes n}\otimes I_{B^n}\big)\, .
\ee
Similarly, compressing $M_{r,n} = \Tr_{B^n}R_{r,n}$ gives
\bb
P_S^{\otimes n}M_{r,n}P_S^{\otimes n} = W^{\otimes n}G(W^\dag)^{\otimes n},
\ee
where
\bb
G \coloneqq \E_C \left[ (CC^\dag)^{\otimes n} \right] .
\ee

We now consider the normalisation term. The law of $V$ is invariant under $V\mapsto U_A V$, so $M_{r,n}$ commutes with $U_A^{\otimes n}$ for every unitary $U_A$ on $A$. By Lemma~\ref{lemma_commutant}, $M_{r,n}$ commutes with every permutationally symmetric operator on $A^n$. In particular,
\bb
\big[M_{r,n},P_S^{\otimes n}\big] = 0\, .
\ee
Thus, $M_{r,n}$ is block diagonal with respect to $S^{\otimes n}\oplus(S^{\otimes n})^\perp$. Therefore its inverse square root
is block diagonal in the same decomposition. On the first block, i.e.\ on $S^{\otimes n}$, the block of $M_{r,n}$ is
\bb
W^{\otimes n}G(W^\dag)^{\otimes n}\, .
\ee
Hence, the corresponding block of $M_{r,n}^{-1/2}$ is
\bb
W^{\otimes n}G^{-1/2}(W^\dag)^{\otimes n}\, .
\ee
Here $G$ is invertible on $\mathcal H_B^{\otimes n}$. Indeed, for every $0\neq \ket{\xi} \in \mathcal{H}_B^{\otimes n}$,
\bb
\braket{\xi|G|\xi} = \E_C \left\|(C^\dag)^{\otimes n}\ket{\xi}\right\|^2 > 0\, ,
\ee
because $C=W^\dagger V$ is invertible with non-zero probability. Equivalently,
\bb
P_S^{\otimes n}M_{r,n}^{-1/2}P_S^{\otimes n} = W^{\otimes n} G^{-1/2}(W^\dag)^{\otimes n}\, .
\ee
Substituting the definition of $\widehat{R}_{r,n}$, and then using the two compression identities above, we obtain
\bb
&\big(P_S^{\otimes n}\otimes I_{B^n}\big)\, \widehat{R}_{r,n}\,\big(P_S^{\otimes n}\otimes I_{B^n}\big) \\
&= \big(P_S^{\otimes n}\otimes I_{B^n}\big)\,\big(M_{r,n}^{-1/2}\otimes I_{B^n}\big)\,R_{r,n}\, \big(M_{r,n}^{-1/2}\otimes I_{B^n}\big)\,\big(P_S^{\otimes n}\otimes I_{B^n}\big) \\
&= \big(P_S^{\otimes n}M_{r,n}^{-1/2}P_S^{\otimes n}\otimes I_{B^n}\big)\, \big(P_S^{\otimes n}\otimes I_{B^n}\big)\, R_{r,n}\,\big(P_S^{\otimes n}\otimes I_{B^n}\big)\,\big(P_S^{\otimes n}M_{r,n}^{-1/2}P_S^{\otimes n}\otimes I_{B^n}\big) \\
&= \big(W^{\otimes n}G^{-1/2}(W^\dagger)^{\otimes n}\otimes I_{B^n}\big)\,\big(W^{\otimes n}\otimes I_{B^n}\big)\, \E_C \left[ \ketbra{\Gamma_C}^{\otimes n} \right] \\
&\quad \cdot \big((W^\dagger)^{\otimes n}\otimes I_{B^n}\big)\, \big(W^{\otimes n}G^{-1/2}(W^\dagger)^{\otimes n}\otimes I_{B^n}\big) \\
&= \big(W^{\otimes n}\otimes I_{B^n}\big)\, \big(G^{-1/2}\otimes I_{B^n}\big)\, \E_C \left[ \ketbra{\Gamma_C}^{\otimes n} \right] \big(G^{-1/2}\otimes I_{B^n}\big)\,\big((W^\dagger)^{\otimes n}\otimes I_{B^n}\big)\, .
\ee
It remains to identify the middle operator. To this end, define the completely positive map
\bb
\Phi(X) \coloneqq \E_C \left[ C^{\otimes n}X(C^\dagger)^{\otimes n} \right] .
\ee
Then $G=\Phi(I_{B^n})$. Define also
\bb
\widetilde\Phi(X) \coloneqq G^{-1/2}\,\Phi(X)\,G^{-1/2}\, .
\ee

We claim that $\widetilde\Phi$ is the Haar twirl:
\bb
\widetilde\Phi(X) = \E_{U\in\mathsf U(r)} \left[ U^{\otimes n}X(U^\dagger)^{\otimes n} \right].
\ee
Indeed, the distribution of $C$ is bi-unitarily invariant, meaning that $U_1CU_2$ has the same distribution as $C$, for all $U_1,U_2\in\mathsf U(r)$. Hence, for every $U\in\mathsf U(r)$,
\bb
U^{\otimes n}\, \widetilde{\Phi}(X)\,(U^\dagger)^{\otimes n} = \widetilde{\Phi}(X)\, , \qquad \widetilde{\Phi}\left( U^{\otimes n}X(U^\dagger)^{\otimes n} \right) = \widetilde{\Phi}(X)\, .
\ee
Thus the output of $\widetilde\Phi$ is invariant, and $\widetilde\Phi$ only depends on the Haar average of its input.

Let $Y$ be invariant under conjugation by \(U^{\otimes n}\). Then \(Y\)
commutes with \(U^{\otimes n}\) for every \(U\). By
Lemma~\ref{lemma_commutant}, \(Y\) commutes with every permutationally symmetric
operator on \(\mathcal H_B^{\otimes n}\). Since \(C^{\otimes n}\) and
\((CC^\dagger)^{\otimes n}\) are permutationally symmetric,
\[
    [Y,C^{\otimes n}]=0,
    \qquad
    [Y,(CC^\dagger)^{\otimes n}]=0.
\]
Therefore
\bb
\begin{aligned}
    \Phi(Y)
    &=
    \mathbb E_C
    \left[ C^{\otimes n}Y(C^\dagger)^{\otimes n}
    \right]  = Y\,
    \mathbb E_C  \left[  (CC^\dagger)^{\otimes n} \right]  =  YG.
\end{aligned}
\ee
Moreover, the same commutation relation also gives \(YG=GY\), and hence
\bb
    \widetilde\Phi(Y) = G^{-1/2}YGG^{-1/2} = Y.
\ee
Thus \(\widetilde\Phi\) fixes every invariant operator.

For a general \(X\), define its Haar average
\bb
    \overline X \coloneqq  \mathbb E_{U\in\mathsf U(r)}
    \left[   U^{\otimes n}X(U^\dagger)^{\otimes n}  \right].
\ee
By the covariance identities above,
\bb
    \widetilde\Phi(X)=\widetilde\Phi(\overline X).
\ee
But \(\overline X\) is invariant, so the previous paragraph gives
\bb
    \widetilde\Phi(\overline X)=\overline X.
\ee
Together, these identities prove the claim.

Passing to Choi operators, we get
\bb
    (G^{-1/2}\otimes I_{B^n})
    \mathbb E_C
    \left[
        \left(
        \left|\Gamma_C\right\rangle
        \!\left\langle\Gamma_C\right|
        \right)^{\otimes n}
    \right]
    (G^{-1/2}\otimes I_{B^n})
    =
    \mathbb E_{U\in\mathsf U(r)}
    \left[
        \left(
        \left|\Gamma_U\right\rangle
        \!\left\langle\Gamma_U\right|
        \right)^{\otimes n}
    \right].
\ee
Substituting this into the previous expression gives
\bb
\begin{aligned}
& (P_S^{\otimes n}\otimes I_{B^n})
    \widehat R_{r,n}
    (P_S^{\otimes n}\otimes I_{B^n})                                      \\
&\quad =
    (W^{\otimes n}\otimes I_{B^n})
    \mathbb E_{U\in\mathsf U(r)}
    \left[
        \left(
        \left|\Gamma_U\right\rangle
        \!\left\langle\Gamma_U\right|
        \right)^{\otimes n}
    \right]
    ((W^\dagger)^{\otimes n}\otimes I_{B^n})                                \\
&\quad =
    \mathbb E_{U\in\mathsf U(r)}
    \left[
        \left(
        \left|\Gamma_{WU}\right\rangle
        \!\left\langle\Gamma_{WU}\right|
        \right)^{\otimes n}
    \right].
\end{aligned}
\ee
This is exactly the desired identity.
\end{proof}

We also need the same commutation trick used in the main text.

\begin{lemma}[(Commutation trick)]
\label{lem:rank-r-commutation}
For every state \(\rho_A\),
\bb
    \sqrt{\widehat R_{r,n}}
    (\rho_A^{\otimes n}\otimes I_{B^n})
    \sqrt{\widehat R_{r,n}}
    =
    (\sqrt{\rho_A}^{\otimes n}\otimes I_{B^n})
    \widehat R_{r,n}
    (\sqrt{\rho_A}^{\otimes n}\otimes I_{B^n}).
\ee
\end{lemma}

\begin{proof}
The law of \(V\) is invariant under \(V\mapsto U_AV\). Hence \(R_{r,n}\)
commutes with \(U_A^{\otimes n}\otimes I_{B^n}\) for every unitary \(U_A\).
Taking the partial trace over \(B^n\), the same invariance gives that
\(M_{r,n}\) commutes with \(U_A^{\otimes n}\). Therefore \(M_{r,n}^{-1/2}\) also
commutes with \(U_A^{\otimes n}\). Since \(\widehat R_{r,n}\) is obtained from
\(R_{r,n}\) by conjugating with \(M_{r,n}^{-1/2}\otimes I_{B^n}\), it follows
that \(\widehat R_{r,n}\) commutes with
\(U_A^{\otimes n}\otimes I_{B^n}\) for every \(U_A\).

By Lemma~\ref{lemma_commutant}, the invariance of \(\widehat R_{r,n}\) under
\(U_A^{\otimes n}\otimes I_{B^n}\), for every unitary \(U_A\) on \(A\), implies
that \(\widehat R_{r,n}\) commutes with
\(X_{A^n}\otimes I_{B^n}\) for every permutationally symmetric operator
\(X_{A^n}\) on \(A^n\). In particular, it commutes with both
\(\rho_A^{\otimes n}\otimes I_{B^n}\) and
\(\sqrt{\rho_A}^{\otimes n}\otimes I_{B^n}\). Consequently
\(\sqrt{\widehat R_{r,n}}\) also commutes with
\(\rho_A^{\otimes n}\otimes I_{B^n}\), and hence
\bb
    \sqrt{\widehat R_{r,n}}
    (\rho_A^{\otimes n}\otimes I_{B^n})
    \sqrt{\widehat R_{r,n}}
    =
    (\sqrt{\rho_A}^{\otimes n}\otimes I_{B^n})
    \widehat R_{r,n}
    (\sqrt{\rho_A}^{\otimes n}\otimes I_{B^n}) .
\ee
\end{proof}

We can now prove the finite-rank version of the random purification channel.

\begin{thm}[(Finite-rank random purification)]
\label{thm:finite-rank-rpc}
For every \(n\ge1\), the channel \(\Lambda_r^{(n)}\) satisfies the following
property. Let \(\rho_A\) be a state on \(A\) such that
\bb
    \operatorname{rank}\rho_A\le r.
\ee
Let \(|\psi_\rho\rangle\) be any purification of \(\rho_A\) on \(AB\), where
\(\dim B=r\), and write
\(\psi_\rho=|\psi_\rho\rangle\!\langle\psi_\rho|\). Then
\bb
    \Lambda_r^{(n)}(\rho_A^{\otimes n})
    =
    \mathbb E_{U_B\in\mathsf U(r)}
    \left[
        \bigl(
        (I_A\otimes U_B)
        \psi_\rho
        (I_A\otimes U_B^\dagger)
        \bigr)^{\otimes n}
    \right].
\ee
\end{thm}

\begin{proof}
It is enough to prove the claim for one fixed purification of \(\rho_A\).
Indeed, all purifications on \(AB\) differ by a unitary on \(B\), and the Haar
average in the statement is unchanged by such a unitary.

Let \(S\) be an \(r\)-dimensional subspace containing
\(\operatorname{supp}\rho_A\). Choose an isometry \(W:\mathcal H_B\to\mathcal H_A\)
with image \(S\), and define
\bb
    \left|\psi_\rho\right\rangle
    \coloneqq 
    (\sqrt{\rho_A}\otimes I_B)\left|\Gamma_W\right\rangle .
\ee
This is a purification of \(\rho_A\). Indeed, since \(WW^\dagger=P_S\) and
\(\rho_A\) is supported on \(S\),
\bb
    \operatorname{Tr}_B
    \left|\psi_\rho\right\rangle\!\left\langle\psi_\rho\right|
    =
    \sqrt{\rho_A}\,WW^\dagger\sqrt{\rho_A}
    =
    \sqrt{\rho_A}\,P_S\sqrt{\rho_A}
    =
    \rho_A.
\ee
Also,
\bb
    \rho_A^{\otimes n}
    =
    P_S^{\otimes n}\rho_A^{\otimes n}P_S^{\otimes n}.
\ee

We first note that the second term in the definition of \(\Lambda_r^{(n)}\)
vanishes on \(\rho_A^{\otimes n}\). It is enough to show that
\bb
    P_S^{\otimes n}\le \Pi_{r,n}.
\ee
Let \(0\neq \xi\in S^{\otimes n}\). Write \(\xi=W^{\otimes n}\eta\) for some
\(0\neq\eta\in\mathcal H_B^{\otimes n}\). With the notation
\(C=W^\dagger V\) from Lemma~\ref{lem:rank-r-compression}, we have
\bb
    \langle \xi|M_{r,n}|\xi\rangle
    &=
    \langle \eta|
    G
    |\eta\rangle  =
    \mathbb E_C
    \left[
        \left\|
        (C^\dagger)^{\otimes n}\eta
        \right\|^2
    \right]
    >0.
\ee
The strict positivity follows because \(C=W^\dagger V\) is invertible with
nonzero probability. Hence \(M_{r,n}\) is strictly positive on \(S^{\otimes n}\),
which proves \(P_S^{\otimes n}\le \Pi_{r,n}\). Combining this with the support
condition on \(\rho_A\), we get
\bb
    (I_{A^n}-\Pi_{r,n})\rho_A^{\otimes n}=0.
\ee
Thus the second term in the definition of \(\Lambda_r^{(n)}\) does not
contribute on \(\rho_A^{\otimes n}\).

By Lemma~\ref{lem:rank-r-commutation},
\bb
    \Lambda_r^{(n)}(\rho_A^{\otimes n})
    =
    (\sqrt{\rho_A}^{\otimes n}\otimes I_{B^n})
    \widehat R_{r,n}
    (\sqrt{\rho_A}^{\otimes n}\otimes I_{B^n}).
\ee
Since \(\rho_A\) is supported on \(S\), we can insert \(P_S^{\otimes n}\) on both
sides of \(\widehat R_{r,n}\). Using Lemma~\ref{lem:rank-r-compression}, we
obtain
\bb
\begin{aligned}
    \Lambda_r^{(n)}(\rho_A^{\otimes n})
    =
    \mathbb E_{U\in\mathsf U(r)}
    \left[
        \left(
        (\sqrt{\rho_A}\otimes I_B)
        \left|\Gamma_{WU}\right\rangle
        \!\left\langle\Gamma_{WU}\right|
        (\sqrt{\rho_A}\otimes I_B)
        \right)^{\otimes n}
    \right].
\end{aligned}
\ee

It remains only to rewrite the state inside the expectation. By the transpose
trick,
\bb
\begin{aligned}
    \left|\Gamma_{WU}\right\rangle
    &=
    (WU\otimes I_B)\left|\Gamma_r\right\rangle        \\
    &=
    (W\otimes I_B)(U\otimes I_B)\left|\Gamma_r\right\rangle \\
    &=
    (W\otimes I_B)(I_B\otimes U^T)\left|\Gamma_r\right\rangle \\
    &=
    (I_A\otimes U^T)\left|\Gamma_W\right\rangle .
\end{aligned}
\ee
Therefore, using the definition of \(|\psi_\rho\rangle\),
\bb
\begin{aligned}
    (\sqrt{\rho_A}\otimes I_B)\left|\Gamma_{WU}\right\rangle
    &=
    (\sqrt{\rho_A}\otimes I_B)
    (I_A\otimes U^T)\left|\Gamma_W\right\rangle       \\
    &=
    (I_A\otimes U^T)
    (\sqrt{\rho_A}\otimes I_B)\left|\Gamma_W\right\rangle \\
    &=
    (I_A\otimes U^T)\left|\psi_\rho\right\rangle .
\end{aligned}
\ee
Substituting this into the previous identity gives
\bb
    \Lambda_r^{(n)}(\rho_A^{\otimes n})
    =
    \mathbb E_{U\in\mathsf U(r)}
    \left[
        \bigl(
        (I_A\otimes U^T)
        \psi_\rho
        (I_A\otimes \overline U)
        \bigr)^{\otimes n}
    \right],
\ee
where \(\psi_\rho=|\psi_\rho\rangle\!\langle\psi_\rho|\). Since \(U^T\) is
Haar-distributed whenever \(U\) is Haar-distributed, we may rename \(U^T\) as
\(U_B\). This gives the desired identity.
\end{proof}

\begin{rem}[(Reduction to the full-rank construction)]
If \(r=d\), every isometry \(V:\mathcal H_B\to\mathcal H_A\) is a unitary. Hence
\bb
    M_{d,n}
    =
    \mathbb E_V
    \left[
        (VV^\dagger)^{\otimes n}
    \right]
    =
    I_{A^n}.
\ee
Consequently,
\bb
    \widehat R_{d,n}
    =
    R_{d,n}.
\ee
Moreover, when \(r=d\), the random isometry \(V:\mathcal H_B\to\mathcal H_A\)
is just a Haar-random unitary identifying \(B\) with \(A\). Therefore \(R_{d,n}\)
is exactly the random maximally entangled operator \(R_n\) of the main text.
Consequently,
\bb
    \Lambda_d^{(n)}(X)
    =
    \sqrt{R_n}(X\otimes I_{B^n})\sqrt{R_n},
\ee
which is precisely the full-rank random purification channel
\(\Lambda^{(n)}\) defined in the main text.
\end{rem}

\end{document}